\documentclass[11pt]{article}
\usepackage[margin=1in]{geometry}	%\usepackage{blindtext}
    \usepackage{cite}
	\usepackage{graphicx}
\usepackage{amsmath}
\usepackage{algorithm}
\usepackage{algpseudocode}

	\usepackage{amssymb}
	\usepackage{color}
	\usepackage[mathscr]{euscript}
	\usepackage{fixmath}
		\usepackage[toc,page]{appendix}

	\usepackage{mdframed}
	\usepackage{tkz-graph}
	 \usepackage{relsize}
	 \usepackage{subcaption}
	 \usepackage{bm}
	 \usepackage{siunitx}
	 \usepackage{balance}
	 	 \usepackage{empheq}

	\newcommand*{\QEDA}{\hfill\ensuremath{\blacksquare}}

	\DeclareMathOperator*{\minimize}{minimize}
\DeclareMathOperator*{\st}{subject~to}

\newtheorem{theorem}{Theorem}
\newtheorem{lemma}{Lemma}
\newtheorem{proposition}{Proposition}
\newtheorem{corollary}{Corollary}

\newtheorem{definition}{Definition} 
\newtheorem{proof}{Proof}

\newtheorem{example}{Example}

	\DeclareMathAlphabet\mathbfcal{OMS}{cmsy}{b}{n}

\begin{document}

	\title{\huge Unified Approach to Convex Robust Distributed Control given Arbitrary Information Structures}

	\author{Luca~Furieri,~
	        Maryam Kamgarpour~  \thanks{This research was gratefully funded by the European Union ERC Starting Grant CONENE. The authors are with the Automatic Control Laboratory, Department of Information Technology and Electrical Engineering, ETH Zurich, Switzerland. e-mails: {\tt\small \{furieril, mkamgar\}@control.ee.ethz.ch}}}% <-this % stops a space
% <-this % stops a space

	\maketitle

	%\IEEEpeerreviewmaketitle

	\begin{abstract}
	We consider the problem of computing optimal linear control policies for linear systems  in finite-horizon. The states and the inputs are required to remain inside pre-specified safety sets at all times despite unknown disturbances. In this technical note, we focus on the requirement  that the control policy is distributed, in the sense that it can only be based on partial information about the history of the outputs. It is well-known that when a condition denoted as Quadratic Invariance (QI) holds, the optimal distributed control policy can be computed in a tractable way. 
	%However, the QI condition can be challenging to be tested in practice in its general formulation, as it must hold for all the infinite controllers lying in a subspace.
%	 Simple conditions equivalent to QI have been proposed for some specific information structures. 
Our goal is to unify and generalize the class of information structures over which quadratic invariance is equivalent to a test over finitely many binary matrices. The test we propose certifies convexity of the output-feedback distributed control problem in finite-horizon given any arbitrarily defined information structure, including the case of time varying communication networks and forgetting mechanisms. Furthermore, the framework we consider allows for including polytopic constraints on the states and the inputs in a natural way, without affecting convexity.

	\end{abstract}	
\section{Introduction}
%\luca{Introduction should be shortened by approx. 1 column}
%MOTIVATION
Critical emerging large-scale systems, such as the electric power grid, autonomous vehicles, the internet and financial systems feature autonomous and interacting decision making agents.  Due to geographic distance or privacy concerns, the agents can only base their decisions on  local partial information. This lack of full information has a significant impact on the design of optimal decisions.

It is well-known that the design of optimal control policies given an information structure can be extremely challenging even in simple cases \cite{Witsenhausen}. Several instances of this problem were shown to be intractable \cite{blondel2000survey, papadimitriou1986intractable}.  Significant work has thus been directed towards identifying special cases of this problem for which efficient algorithms can be derived. The work in \cite{ho1972team} showed that when the information structure is Partially Nested (PN) the linear-quadratic Gaussian stochastic control problem can be reduced to a static problem, whose optimal solution is linear. More recently, the work  \cite{voulgaris2001convex} established convexity of the problem of optimal disturbance rejection in broader classes of controlled systems, such as those being ``nested, chained, hierarchical, delayed interaction and communication, or symmetric'' \cite{voulgaris2001convex}.
 In \cite{bamieh2005convex}, optimal distributed controller design for spatially invariant systems was shown to be convex given ``funnel causality'' of the information structure.
%dynamically decoupled and spatially distributed systems \cite{SpatiallyDistributed}, \cite{SpatiallyDistributed2}.
 %An important family of information structures, defined as \emph{partially nested}, was introduced in \cite{PartiallyNested}, which allows to shape optimal controllers via convex programming. 
%The \emph{one-step delay information sharing pattern} problem studied in \cite{Voulgaris} also has this property.

The cases \cite{ho1972team,voulgaris2001convex,bamieh2005convex} were unified in the celebrated works \cite{rotkowitz2005tractable,rotkowitz2006characterization,rotkowitz2008information,QIconvexity}, where the authors established  necessary and sufficient conditions on the plant and the  information structure for convexity of optimal control referred to as \emph{quadratic invariance} (QI).  
%The derivation of QI for convex design of optimal distributed controllers is based on the Youla parametrization  \cite{youla1976modern}. 
 Quadratic invariance revealed that full sensor information is needed when all the states dynamically influence each other. Such limitation encouraged including a  communication network in the distributed control system design to restore convexity. The works \cite{wang2014localized,wang2016system} proposed a system level approach, which is based on allowing controllers to share past control inputs, thus improving overall convexity and performance. Input sharing might raise privacy concerns in competitive scenarios where controllers need to keep their control strategies unknown to others. Instead, \cite{rotkowitz2010convexity} considered a time-invariant communication network to share local output measurements with fixed delays. The authors of  \cite{rotkowitz2010convexity}  used the concepts of propagation and transmission delays to derive algebraic conditions that are equivalent to QI for these information structures.

Once a convex problem is obtained under the QI condition, one possible approach is to compute optimal controllers in the form of transfer functions by approximating the corresponding infinite-dimensional convex optimization problem \cite{Alavian}.  In addition to relying on approximations, these techniques might lead  to  controllers whose order is too high for practical implementation. Moreover, operations at the transfer function level can be numerically unstable \cite{shah2013cal} and may not offer sufficient insight into the dynamics of the internal states of controllers. To deal with these  issues, another approach is to explicitly compute distributed optimal controllers within state space representations by using dynamic programming. This approach led to several works, each addressing a specific class of dynamical systems and information structures as follows. The authors in \cite{shah2013cal} considered the case of poset-causal systems given  that controllers can only access partial sensor measurements at each time  and showed that  the problem reduces to a set of decoupled centralized problems.  The case of output sharing between controllers with one-timestep delays was solved in \cite{kurtaran1974linear,sandell1974solution,yoshikawa1975dynamic}.  Combining sparse sensor measurements and fixed delays in the communication between controllers was considered   in \cite{lamperski2015optimal}. In \cite{matni2014optimal}, the authors solved the more general case of time-varying delays given a special class of systems consisting of two interconnected plants. % in reinterpreted the results in \cite{rotkowitz2005tractable,rotkowitz2006characterization,rotkowitz2008information} in terms of partially ordered sets (posets) \cite{dushnik1941partially}. 

% Control inputs typically represent physical actuation applied to the system, hence  cannot  exceed certain values. State variables are linked to physical quantities of interest and should always remain within given boundaries for safety reasons.
In several control scenarios of practical interest, it is important to ensure that states and control inputs remain inside specified safety sets at all times. However, the past works on optimal distributed control based on convex programming  \cite{voulgaris2001convex,bamieh2005convex,rotkowitz2005tractable,rotkowitz2006characterization,rotkowitz2008information
}  and on explicit dynamic programming \cite{shah2013cal,kurtaran1974linear,sandell1974solution,yoshikawa1975dynamic,lamperski2015optimal,matni2014optimal} did not address the case where constraints on the state and input trajectories must be enforced. %Specifically, it can be challenging to adapt the approaches \cite{shah2013cal,kurtaran1974linear,sandell1974solution,yoshikawa1975dynamic,lamperski2015optimal,matni2014optimal}  to constrained control, thus making the use of convex programming preferable in several cases.  
%Indeed, modern control structures such as the well established receding horizon control, address constrained dynamical systems. 
%Given the importance of constraints on states and inputs, significant effort  has been recently directed towards including them into optimal control problems given an information structure. 
A technique to include state and input constraints is distributed model predictive control \cite{Zeilinger, Rantzer}. Typically,  this approach involves dividing a large-scale system into subsystems that  communicate in order to minimize a given cost function, while being robust to  the dynamical couplings. Convexity of the resulting optimization problems is obtained at the cost of potential conservativeness by optimizing over open-loop control policies. However, repeating the optimization at every time step mitigates this suboptimality.  A different approach is to consider closed-loop predictions that allow for longer prediction horizons without running into infeasibility of the optimization problem \cite{bemporad1998reducing}. Following this line of work,  \cite{furieri2017robust} has recently considered the problem of convex robust distributed control with closed-loop predictions given a fixed sensing topology. In this technical note, we adopt the closed-loop prediction approach.%, and has derived necessary and sufficient conditions for convexity of the problem based on exploiting disturbance-feedback control policies and QI related results. 

The  literature review above reveals that, when distributed control is considered, different instances of dynamical systems \cite{voulgaris2001convex,shah2013cal}, time-invariant and time-varying information structures \cite{shah2013cal,kurtaran1974linear,sandell1974solution,yoshikawa1975dynamic,lamperski2015optimal, matni2014optimal}, and the inclusion of state and input constraints \cite{Zeilinger, Rantzer, bemporad1998reducing, furieri2017robust} are treated as separate challenges, to be addressed with a variety of different approaches and mathematical tools. The contributions of this paper are as follows. First, we provide a unified and generalized approach to convex design of optimal distributed controllers. We show that sparsity constraints, delay constraints, arbitrary time-varying sensing and communication networks and forgetting mechanisms can all be combined within the framework of robust distributed control in finite-horizon. Second, we characterize convexity given any arbitrary information structure as per above, in terms of a test on finitely many  binary matrices.  The test we  propose generalizes those proposed in \cite{rotkowitz2006characterization, furieri2017robust, rotkowitz2010convexity,abara2018quadratic} for particular information structures, in the sense that they can all be alternatively proved as direct corollaries of our main result. To show generality of our approach, we provide an example of a quadratically invariant information structure which was not captured in the previous works. Third, our framework allows for including polytopic constraints on the state and input trajectories in a natural and forthright way.

The paper is structured as follows. In Section \ref{se:preliminaries}, after introducing the problem under study we explain how to encode arbitrary information structures using binary matrices whose dimensions grow with the prediction horizon. In Section~\ref{se:convexity} we present our main result on a generalized binary test  for quadratic invariance, and show an example of a QI information structure which was not captured in previous work.   %which are based on verifying a finite number of binary inequalities. %This allows for characterizing sufficient and necessary conditions for convexity of the robust distributed control problem in the presence of generic information structures by exploiting QI arguments on high-dimensional sparsity subspaces. %, and show that results from \cite{rotkowitz2012nearest} can be naturally applied to compute non-trivial lower on the minimum cost when convexity does not hold.
   In Section \ref{sec:dyn_complete}, we specialize our result to the case of combining a fixed sensing topology and a fixed communication network. In doing so, we present alternative proofs to the results of \cite{furieri2017robust,rotkowitz2010convexity} as direct corollaries of our main theorem. These new proofs offer additional insight on the results of \cite{furieri2017robust,rotkowitz2010convexity} in terms of the minimal set of inequalities that must be verified for convexity.%These new proofs allow for deriving additional insight on the results of \cite{furieri2017robust,rotkowitz2010convexity}. % the necessary and sufficient conditions for convexity to the realistinfic case of combining a fixed sensing topology and a fixed communication topology with one-time-step delays, as similarly considered in \cite{rotkowitz2010convexity}. We show that our framework allows for a clean proof of convexity based on few conditions on  low-dimensional binary matrices, without the need of defining the concepts of propagation and transmission delays. 
\section{Preliminaries}
\label{se:preliminaries}

In this section, we introduce some notation on sparsity structures, and  present the problem of distributed optimal control in finite-horizon with hard constraints on the states and the inputs.  %We then present our main  result on characterizing the convexity of this control problem in terms of a finite number of conditions over binary matrices, for any arbitrarily defined information structure.% thus generalizing the analysis conducted in \cite{rotkowitz2010parametrization} to consider arbitrary information structures.% The notion of \emph{quadratic invariance} is briefly discussed as well.

\subsection{Notation and sparsity structures}	
 We use $\mathbb{R}$ and $\mathbb{N}$ to  denote the sets of real numbers and  positive integers respectively. For any $n \in \mathbb{N}$, we use $\mathbb{N}_{[a,b]}\subseteq \mathbb{N}$ to denote the set of integers from $a$ to $b$.  The $(i,j)$-th element in a matrix $Y \in \mathbb{R}^{m \times n}$ is referred to as $Y(i,j)$.  Throughout the paper we indicate the $i$-th component of a vector $v \in \mathbb{R}^{n}$ as $v^i \in \mathbb{R}$.  We use $I_n$ to denote the identity matrix of size $n \times n$ and  $0_{m \times n}$ to denote the zero matrix of size $m \times n$.  The sparsity structure of a matrix can be conveniently represented by a binary matrix. A binary matrix is a matrix with entries from the set $\{0,1\}$, and we use $\{0,1\}^{m \times n}$ to denote the set of $m \times n$ binary matrices.
%We denote the set of $n \times n$ symmetric binary matrices as $\mathbb{S}_n^{B}$
Given a binary matrix $X \in \{0,1\}^{m \times n}$, we define the sparsity subspace $\text{Sparse}(X) \subseteq \mathbb{R}^{m \times n}$ as
\begin{equation*}
\text{Sparse}(X) := \{Y \in \mathbb{R}^{m \times n}\mid Y(i,j) = 0, \; \text{if}\; X(i,j)=0, \forall i,j\}.
\end{equation*}
Similarly, given $Y\in \mathbb{R}^{m \times n}$, we define a binary matrix $X\, := \, \text{Struct}(Y)$ as
$$
    X(i,j)=\begin{cases}
	   1 & \text{if}\;Y(i,j) \neq 0\,,\\
	   0 & \text{otherwise}\,.
	\end{cases}
$$
%For example, given  a binary matrix $X \in \{0,1\}^{2 \times 3}$ and a matrix $Y \in \mathbb{R}^{2 \times 3}$ as follows
%\begin{equation} \label{E:BinaryExample}
%    X = \begin{bmatrix}
%        0 & 1 & 0 \\ 1 & 1 & 1
%    \end{bmatrix}\,, \quad Y = \begin{bmatrix}
%        0 & 2 & 0 \\ 3 & 4 & 2
%    \end{bmatrix},
%\end{equation}
%we have $Y \in \text{Sparse}(X)$, and $X = \text{Struct}(Y)$.

 Let $X, \hat{X} \in \{0,1\}^{m \times n}$ and $Z \in \{0,1\}^{n \times p}$ %, and $Z \in \{0,1\}^{c \times c}$
be binary matrices. Throughout the paper, we adopt the following conventions:
$$
    \begin{aligned}
        X + \hat{X} &:= \text{Struct}(X + \hat{X}), \\
        XZ&:=\text{Struct}(XZ). \\
       % Z^r&:=\text{Struct}(Z^r).  % this definition is redundent since it follows from the second equation above.
    \end{aligned}
$$
We state that $X \leq \hat{X}$ if and only if $X(i,j)\leq \hat{X}(i,j)\;\forall i,j$, and $X < \hat{X}$ if and only if $X \leq \hat{X}$ and there exist indices $i,j$ such that $X(i,j)<\hat{X}(i,j)$. {Also, we denote} $X \nleq \hat{X}$ if and only if there exist indices $i,j$ such that $X(i,j)>\hat{X}(i,j)$.
%	
%	 Given a binary matrix $ X \in \{0,1\}^{a \times b}$ we denote its cardinality, \emph{i.e.}, the total number of nonzero entries, as
%	 \begin{equation*}
%	  \|X\|_0=\sum_{i=1}^a \sum_{j=1}^b X_{ij}. % \in \mathbb{N}\,,
%	 \end{equation*}
%Considering the following binary matrices 
%$$
%    {X}_1 = \begin{bmatrix}
%        0 & 1 & 0 \\ 1 & 0 & 1
%    \end{bmatrix}, {X}_2 = \begin{bmatrix}
%        1 & 1 & 0 \\ 1 & 0 & 1
%    \end{bmatrix},
%$$
%and the binary matrix $X$ in~\eqref{E:BinaryExample}, we have ${X}_1 < X, {X}_2 \nleq X$ and ${X}_1 + X = X$. The cardinality of these matrices are $\|X\|_0 = 4, \|X_1\|_0 = 3$ and $ \|X_2\|_0 = 4$, respectively.
\subsection{Problem formulation}
\label{se:setup}
	We consider a discrete time system
	\begin{equation}
	\label{eq:system}
	\begin{aligned}
	&x_{k+1}=Ax_k+Bu_k+Dw_k\,, \quad y_k=Cx_k+Hw_k\,,\\
	\end{aligned}
	\end{equation}
	where $k\geq 0$ is an integer number, $y_k \in \mathbb{R}^p$, $u_k \in \mathbb{R}^m$, $w_k\in \mathcal{W}$ and $\mathcal{W} \subseteq \mathbb{R}^n$ is a closed and bounded set of possible disturbances.  The system starts from a known initial condition $x_0\in \mathbb{R}^n$. Let us define a prediction horizon of length $N \in \mathbb{N}$. Our goal is to minimize a cost function of the history of states and inputs $J(x_0,\cdots,x_N,u_0,\cdots,u_{N-1})$. Furthermore, the states and inputs need to satisfy
	\begin{equation}
	\label{eq:constraints_state_input}
	\begin{aligned}
	&\begin{bmatrix}x_k\\u_k\end{bmatrix} \in \Gamma\subseteq \mathbb{R}^{n+m}\,,\quad x_N \in \mathcal{X}_f \subseteq \mathbb{R}^{n}\,,
	\end{aligned}
	\end{equation}
	 for all $k \in \mathbb{N}_{[0,N-1]}$ and for all realizations of disturbances taken from set $\mathcal{W}$. Each control input can base its decisions on a given subset of the present and past outputs. This subset is denoted as the	\emph{information structure}.  An information structure can be arbitrary, in the sense that controllers may measure, receive, memorize or forget any output measurements at any time.	
%	 may change over time. An information structure is defined as the set
%	\begin{equation}	
%		\label{eq:information_structures_general}
%		\{\mathcal{I}^i_k~\text{ s.t. }~ i\in \mathbb{N}_{[1,m]},~k \in \mathbb{N}_{[0,N-1]}\}\,,
%	\end{equation}	
%where
%	\begin{equation}
%	\label{eq:information_structures_general2}
%	\mathcal{I}^i_k=\left\{y^j_l~\text{ s.t. }~
%		\begin{aligned}
%			&\quad y^j_l\text{ is known to controller }i\text{ at time $k$}\,,\\
%			&\quad  j \in \mathbb{N}_{[1,p]},~ l \in \mathbb{N}_{[0,k]}\,.
%			%&\quad u^{j'}_{l'}\text{ is known to controller }i\text{ at time $k$}\,.
%		\end{aligned}
%	\right\}\,.
%	\end{equation}
% Accordingly, we require that $u^i_k$ is only a function of the outputs in the set $\mathcal{I}^i_k$, for all $i \in \mathbb{N}_{[1,m]}$ and $k \in \mathbb{N}_{[0,N-1]}$.  
	% as defined in what follows. Additionally,  %$u_k=\gamma(y_0,\cdots,y_k),~\gamma:\quad \mathbb{R}^{p(k+1)}\rightarrow \mathbb{R}^{m}$
%	constraints (\ref{eq:constraints_state_input})  must be satisfied, 
%It should be noted that, due to the presence of disturbance in the dynamics (\ref{eq:system}), approaches based on open-loop control may lead to infeasibility for large values of the horizon $N$ \cite{feedback_predictions}.
The search in the class of all output-feedback policies is  intractable.   Hence, a possible approach is to restrict the search to the class of controllers that are affine in the history of the outputs. The time-varying output-feedback affine  control policy is expressed as

		\begin{equation}
	\label{eq:affine_feedback}
	\begin{aligned}
	&u_k=\sum_{j=0}^kL_{k,j}y_j+g_k\,,
	\end{aligned}
	\end{equation}
	for all time instants $k \in \mathbb{N}_{[0,N-1]}$. In (\ref{eq:affine_feedback}), the matrices $L_{k,j} \in \mathbb{R}^{m \times p}$ represent the output-feedback part of the control policy, while the vectors $g_k \in \mathbb{R}^{m}$ represent a feed-forward term that can be synthesized to improve performance.   For every  $j \in \mathbb{N}_{[0,k]}$, we consider  binary matrices $S_{k,j}\in \{0,1\}^{m \times p}$ describing the information known at time $k$ to controllers  about the outputs at time $j\leq k$. Specifically, $S_{k,j}(a,b)=1$ if and only if the $a$-th scalar control input at time $k$ knows the $b$-th scalar output  at time $j$.   A given information structure on the input can thus be enforced as
	\begin{equation}
	\label{eq:affine_feedback_constraints}
	\begin{aligned}
	%&u_k=\sum_{j=0}^kL_{k,j}y_j+g_k,~\hat{u}_k=\sum_{j=0}^kL_{k,j}\hat{y}_j+g_k\,,
	&L_{k,j} \in \text{Sparse}(S_{k,j}) \,,
	& \forall k \in \mathbb{N}_{[0,N-1]},~\forall j \in \mathbb{N}_{[0,k]}\,. 
	%&\forall k \in \mathbb{N}_{[0,N-1]},~\forall  j \in \mathbb{N}_{[0,k]}\,.
	\end{aligned}\
	\end{equation}
\emph{Arbitrary  information structures}:	 The general description we suggest in (\ref{eq:affine_feedback_constraints}) allows for unifying the treatment of  several  classes of information structures considered in the literature. We provide a list below.
\begin{enumerate}
\item Time-invariant sensing topologies described by $S \in \{0,1\}^{m \times p}$ as per \cite{furieri2017robust} can be encoded by selecting $S_{k,j}=S$ for all $k \in \mathbb{N}_{[0,N-1]}$ and $j \in \mathbb{N}_{[0,k]}$.
%\item time-varying sensing topologies by selecting $S_{k,j}=S_k$ for every $j$ given a time $k$.
\item Networked systems with  fixed communication delays between subsystems as considered in \cite{lamperski2015optimal} can be encoded as follows.  Let $d_{a,b} \in \mathbb{N}$ be the time steps it takes  for the information about the $b$-th output to reach the $a$-th input, where the case $d_{a,b}=\infty$ is denoted in \cite{lamperski2015optimal} as a sparsity constraint. Then, for any given time $k$, we  set $S_{k,l}(a,b)=1$ for all $l\in \mathbb{N}_{[0,k-d_{a,b}]}$ and $S_{k,l}(a,b)=0$ for all $l\in \mathbb{N}_{[k-d_{a,b}+1,k]}$.%One simple way to encode this information structure is to select binary matrices $S$ and $Z$ representing sparsity constraints and the communication topology respectively. Then, we pose  $S_{k,j}=Z^{\lfloor (k-j)d^{-1} \rfloor}S$ for every $k$ and $j$ as per (\ref{eq:affine_feedback_constraints}), where $\lfloor \cdot \rfloor$ denotes the integer part.
\item Time-varying delays in a networked control system \cite{matni2014optimal} can be encoded as follows. Let $e_{a,b}(k)$ be the number of time steps it takes at time $k$ for the information about the $b$-th output to reach the $a$-th input.  Then, for any given time $k$, we set $S_{k,l}(a,b)=1$ for each $l \in \mathbb{N}_{[0,k-e_{a,b}(k)]}$ and $S_{k,l}(a,b)=0$ for all $l\in \mathbb{N}_{[k-e_{a,b}(k)+1,k]}$. 
\item Starting from the results of this technical note, the work \cite{abara2018quadratic} has recently characterized convexity for the case of scheduled intermittent observations. See \cite[Lemma~1]{abara2018quadratic} for the appropriate choice of matrices $S_{k,j}$ in terms of time-varying sensing and communication topologies.
\item Random information structures as considered in \cite{ouyang2017stochastic} can be encoded by letting  $S_{k,j}$ be  taken at random from a pre-specified set of possible sparsities. 
\end{enumerate}

	We are now ready to state the optimization problem under study given an arbitrary information structure.
	\begin{alignat*}{3}
	&\text{\textbf{Problem 1}}&&~\\
	 &\minimize&&J(x_0,\cdots,x_N,u_0,\cdots,u_{N-1})\\ 
	 &\st &&(\ref{eq:system}),(\ref{eq:constraints_state_input}) \quad \forall w_k\in \mathcal{W}\,,\\
	&~&& (\ref{eq:affine_feedback}),(\ref{eq:affine_feedback_constraints}) \quad  \forall k \in \mathbb{N}_{[0,N-1]}\,, \forall j \in \mathbb{N}_{[0,k]}.
		 	\end{alignat*}
	%\end{equation}	
	In the problem above, the decision variables are the output-feedback matrices $L_{k,j}$ and the feed-forward vectors $g_k$  as in (\ref{eq:affine_feedback}) for all $k \in \mathbb{N}_{[0,N-1]}$ and $j \in \mathbb{N}_{[0,k]}$. As we optimize over the matrices defining the closed-loop policy (\ref{eq:affine_feedback}), we say that we perform closed-loop predictions. Without any knowledge about the statistics of the realization of $w_k$ at each time in the set $\mathcal{W}$, we consider $J(\cdot)$ to be a function of the disturbance-free state and input trajectories. For computational tractability, we assume that $J(\cdot)$ is a convex function and for simplicity we assume that the sets $\Gamma,~\mathcal{X}_f$ are polytopes  $\Gamma=\{(x,u) \in (\mathbb{R}^{n}, \mathbb{R}^m)\text{ s.t. }Ux+Vu\leq b\}$, where $U \in \mathbb{R}^{s\times n}$, $V \in \mathbb{R}^{s \times m}$ and $b \in \mathbb{R}^s$, and  $\mathcal{X}_f=\{x \in \mathbb{R}^{n}\text{ s.t. }Rx \leq z\}$,
where $R \in \mathbb{R}^{r \times n}$ and	$z \in \mathbb{R}^r$.  Despite these assumptions, it was shown in \cite{Goulart} that Problem~1 is non-convex in general. This is due to the nonlinear propagation of the state feedback matrices through the linear model (\ref{eq:system}).

		\balance

\section{Convexity For Arbitrary Information Structures}
\label{se:convexity}
Our goal is then to find a convex  program which is  equivalent to the intractable Problem~1.  It is convenient to define the vectors of stacked variables as
	\begin{equation*}
	\begin{aligned}
	&\mathbf{x}=\begin{bmatrix}x_0^\mathsf{T}&\cdots&x_N^\mathsf{T}\end{bmatrix}^{\mathsf{T}} \in \mathbb{R}^{n(N+1)}\,,\\
	%&\mathbf{y}=\begin{bmatrix}y_0^\mathsf{T}&\cdots&y_{N-1}^\mathsf{T}\end{bmatrix}^{\mathsf{T}} \in \mathbb{R}^{pN}\,,\\
	&\mathbf{y}=\begin{bmatrix}y_0^\mathsf{T}&\cdots&y_{N}^\mathsf{T}\end{bmatrix}^\mathsf{T} \in \mathbb{R}^{p(N+1)}\,,\\
	 &\mathbf{u}=\begin{bmatrix}u_0^\mathsf{T}&\cdots&u_{N-1}^\mathsf{T}&0_{m \times 1}^\mathsf{T}\end{bmatrix}^\mathsf{T}\in \mathbb{R}^{m(N+1)}\,,\\
	%&\mathbf{\hat{x}}=\begin{bmatrix}\hat{x}_0^{\mathsf{T}}&...&\hat{x}_N^{\mathsf{T}}\end{bmatrix}^{\mathsf{T}}\in \mathbb{R}^{n(N+1)}\,,\\
	%&\mathbf{\hat{u}}=\begin{bmatrix}\hat{u}_0^\mathsf{T}&...&\hat{u}_{N-1}^\mathsf{T}\end{bmatrix}^\mathsf{T}\in \mathbb{R}^{mN}\,,\\
	%&\mathbf{w}=\begin{bmatrix}w_0^\mathsf{T}&...&w_{N-1}^{\mathsf{T}}\end{bmatrix}^{\mathsf{T}}\in \mathbb{R}^{nN}\,,\\
	&\mathbf{w}=\begin{bmatrix}w_0^\mathsf{T}&\cdots&w_{N-1}^{\mathsf{T}}&0_{n \times 1}^\mathsf{T}\end{bmatrix}^{\mathsf{T}}\in \mathbb{R}^{n(N+1)}\,.\\
	\end{aligned}
	\end{equation*}
% As the state $x_{N+1}$ is outside the planning horizon, $u_N$ and $w_N$ can be chosen freely and are only considered for dimensional consistency. We set them to zero for simplicity.
	Equation (\ref{eq:system}) can be succinctly expressed as
	\begin{equation}
	\label{eq:states_stacked}
	\begin{aligned}
	&\mathbf{x}=\mathbf{A}x_0+\mathbf{Bu}+\mathbf{E}_D\mathbf{w}\,, \quad \mathbf{y}=\mathbf{Cx+Hw}\,,\\%, \quad \mathbf{\hat{x}}=\mathbf{B\hat{u}}+\mathbf{\tilde{E}}_D\begin{bmatrix}x_0^\mathsf{T}&0&\cdots&0\end{bmatrix}^\mathsf{T}\,,\\
 %\quad \mathbf{\hat{y}}=\mathbf{C\hat{x}+H\begin{bmatrix}x_0^\mathsf{T}&0&\cdots&0\end{bmatrix}^\mathsf{T}}\,,
	\end{aligned}\
	\end{equation}
where matrices $\mathbf{A}$, $\mathbf{B}$, $\mathbf{E}_D$, $\mathbf{C}$ and $\mathbf{H}$ are defined in Appendix~\ref{app:notation}. Their derivation is straightforward from the recursive application of (\ref{eq:system}).
	%and constraints (\ref{eq:constraints_state_input}), (\ref{eq:constraints_final_state}) and are explicitly written in the appendix.
	Similarly,  considering (\ref{eq:affine_feedback}), the control input can be expressed as	$\mathbf{u}=\mathbf{Ly+g}$,	where $\mathbf{L}\in \mathbb{R}^{m(N+1)\times p(N+1)}$ and $\mathbf{g}\in \mathbb{R}^{m(N+1)}$ are defined in Appendix~\ref{app:notation}.  In order to satisfy (\ref{eq:affine_feedback_constraints}) matrix $\mathbf{L}$ must lie in a subspace $\mathbfcal{S} \subseteq \mathbb{R}^{m(N+1) \times p(N+1)}$, where $\mathbfcal{S}=\text{Sparse}(\mathbf{S})$ and $\mathbf{S}$ stacks the matrices $S_{k,j}$'s as per (\ref{eq:L}) in Appendix~\ref{app:notation}.	
	%The reparametrization is obtained by searching over the class of affine \emph{disturbance-feedback} controllers:
	
	\subsection{Disturbance-feedback parametrization}
It is known that when an information structure is not enforced, parametrizing the controller as a disturbance-feedback affine policy restores tractability of Problem~1 \cite{Goulart}.  Letting $\mathbf{P=C}\mathbf{E}_D+\mathbf{H}$ we define the disturbance-feedback controller
\begin{equation}
\label{eq:dist_feedback}
 \mathbf{u}=\mathbf{QPw}+\mathbf{v}\,.
 \end{equation}
%	\begin{equation}
%	\label{eq:inputs_stacked_df}
%	\mathbf{u}=\mathbf{QPw}+\mathbf{v}\,. \\ % \quad \mathbf{\hat{u}}=\mathbf{M}\begin{bmatrix}x_0^\mathsf{T}&0&\cdots&0\end{bmatrix}^\mathsf{T}+\mathbf{v}\,.
%	\end{equation}
 The decision variable  $\mathbf{Q}\in \mathbb{R}^{m(N+1)\times p(N+1)}$  is causal as per (\ref{eq:L}).
%			\begin{equation}
%	\label{eq:Q_tilde}
%	\begin{aligned}
%	&\mathbf{Q}=\begin{bmatrix}
%	Q_{0,0}&0_{m \times p}&0_{m \times p}&...&0_{m \times p}\\
%	Q_{1,0}&Q_{1,1}&0_{m \times p}&...&0_{m\times p}\\
%	%M_{1,0}&M_{1,1}&...&0&0_{m \times p}\\
%	...&...&...&...&...\\
%	Q_{N-1}&Q_{N-1,1}&...&Q_{N-1,N-1}&0_{m \times p}\\
%	0_{m \times p}&0_{m \times p}&...&0_{m \times p}&0_{m \times p}
%	\end{bmatrix}\,,\\
%	&\mathbf{v}=\begin{bmatrix}v_0^\mathsf{T}&...&v_{N-1}^\mathsf{T}&0_{m \times 1}^\mathsf{T}\end{bmatrix}^\mathsf{T}\,.
%	\end{aligned}
%	\end{equation} 
	It is easy to verify that we can map a disturbance-feedback controller $(\mathbf{Q,v})$ to the unique corresponding output-feedback controller $(\mathbf{L,g})$ and vice versa as follows.
\begin{align}
&\mathbf{L}=\mathbf{Q}(\mathbf{CBQ}+I_{p(N+1)})^{-1}\,,\label{eq:map_Q_to_L}\\
&\mathbf{g}=\mathbf{v}-\mathbf{Q}(\mathbf{CBQ}+I_{p(N+1)})^{-1}(\mathbf{CBv+CA}x_0)\,,\nonumber\\
&\mathbf{Q}=\mathbf{L}(I_{p(N+1)}-\mathbf{CBL})^{-1}\,, \label{eq:map_L_to_Q}\\
&\mathbf{v}=\mathbf{L}(I_{p(N+1)}-\mathbf{CBL})^{-1}(\mathbf{CBg+CA}x_0)+\mathbf{g}\,.\nonumber
\end{align}
 Motivated by \cite{rotkowitz2005tractable,rotkowitz2006characterization,rotkowitz2008information, QIconvexity} we introduce the closed-loop map operator for systems in finite-horizon.
\begin{definition}[Closed-loop map]
Let $\mathbf{X} \in \mathbb{R}^{m(N+1)\times p(N+1)}$ and $\mathbf{Y} \in \mathbb{R}^{p(N+1)\times m(N+1)}$. The closed-loop map function $h:\mathbb{R}^{m(N+1)\times p(N+1)}\times \mathbb{R}^{p(N+1)\times m(N+1)}\longrightarrow ~\mathbb{R}^{m(N+1) \times p(N+1)}$ is defined such that $h(\mathbf{X},\mathbf{Y})=-\mathbf{X}(I_{p(N+1)}-\mathbf{YX})^{-1}$.
%\begin{alignat*}{4}
%h:~&\mathbb{R}^{m(N+1)\times p(N+1)}\times \mathbb{R}^{p(N+1)\times m(N+1)}&&\longrightarrow ~&&&\mathbb{R}^{m(N+1) \times p(N+1)} \\
%  &\qquad \qquad \qquad (\mathbf{X,Y})&&\longmapsto &&&-\mathbf{X}(I_{p(N+1)}-\mathbf{YX})^{-1}\,.
%\end{alignat*}
%The set $h(\mathbfcal{S},\mathbf{CB})$ is defined as
%\begin{equation*}
%\label{eq:h_set_definition}
%h(\mathbfcal{S},\mathbf{CB})=\{h(\mathbf{L,CB}),~\forall \mathbf{L}\in \mathbfcal{S}\}\,.
%\end{equation*}
\end{definition}

With reference to mappings (\ref{eq:map_Q_to_L}) and (\ref{eq:map_L_to_Q}), notice that the operator $h(\cdot)$ maps an output-feedback controller $\mathbf{L}$ to the corresponding disturbance-feedback controller $-\mathbf{Q}$. In particular, the mappings (\ref{eq:map_Q_to_L}) and (\ref{eq:map_L_to_Q}) can be expressed as $\mathbf{L}=h(\mathbf{-Q},\mathbf{CB})$ and $\mathbf{Q}=-h(\mathbf{L},\mathbf{CB})$ respectively.
 %Conditions based on Partial Nestedness (PN) of the information structure were studied in \cite{lin2016performance,lin2017convex}.
% The mappings (\ref{eq:map_Q_to_L}) and (\ref{eq:map_L_to_Q}) do not depend on the matrices $D$ and $H$ and can always be computed. Note that if $H$ and $D$ are invertible it is additionally possible to reconstruct the disturbance at time $k$ by first computing $x_k$ as a function of $x_{k-1}$, $u_{k -1}$ and $w_{k -1}$ through (\ref{eq:system}) and then comparing the measured output $y_k$ with $Cx_k$. 
It is easy to show that when $J(\cdot)$ is convex in the disturbance-free states and inputs, then it does not depend on $\mathbf{Q}$ and it is convex in $\mathbf{v}$. Similar to \cite{Goulart}, the state and input constraints (\ref{eq:constraints_state_input}) are affine in ($\mathbf{Q,v}$) and  can be expressed as
\begin{equation}
\label{eq:affine_constraints}
F\mathbf{v}+\max_{\mathbf{w}\in \mathcal{W}^{N +1}}(F\mathbf{QP}+G)\mathbf{w}\leq c\,,
\end{equation}
 where $F\in \mathbb{R}^{(Ns+r) \times m(N+1)},~G \in \mathbb{R}^{(Ns+r) \times n(N+1)},~c \in \mathbb{R}^{Ns+r}$ are reported in Appendix~\ref{app:notation} for completeness, and $\max_{\mathbf{w}\in \mathcal{W}^{N +1}}(\cdot)$ denotes row-wise maximization. We remark that no assumption on the convexity of the set of disturbances $\mathcal{W}$ is needed for the set of disturbance-feedback controllers $(\mathbf{Q},\mathbf{v})$ satisfying  (\ref{eq:affine_constraints}) to be convex \cite{Goulart}. Nonetheless, a non-convex set of disturbances $\mathcal{W}$ can be substituted by its convex hull without loss of generality or performance \cite[Example 7.1.2.]{bertsekas2003convex}. In this case, computing a policy that satisfies (\ref{eq:affine_constraints}) is easily done by formulating the dual of the maximization problems corresponding to each row of $\max_{\mathbf{w}\in \mathcal{W}^{N +1}}(F\mathbf{QP}+G)\mathbf{w}$. We refer the reader to \cite[Section 4.2]{Goulart} for the specific cases of polytopic and norm-bounded disturbances.
 
 {
 By the discussion above, we conclude that Problem~1 can be equivalently formulated as the following optimization problem:
  \begin{alignat}{3}
	&\textbf{Problem~2} \nonumber\\
	 &\minimize_{\mathbf{Q,v}}&&J(x_0,\mathbf{v}) \nonumber \\ 
	 &\st &&F\mathbf{v}+\max_{\mathbf{w}\in \mathcal{W}^{N +1}}(F\mathbf{QP}+G)\mathbf{w}\leq c\,,\nonumber\\
	 &~&& \mathbf{Q} \in -h(\mathbfcal{S},\mathbf{CB}) \,. \label{eq:sparsity_nonlinear}
	 \end{alignat}
	The optimal solution of Problem~2 can be translated into the optimal solution of Problem~1, and vice-versa, through the mappings (\ref{eq:map_Q_to_L}), (\ref{eq:map_L_to_Q}). However, Problem~2 is still an intractable problem due to  $h(\mathbfcal{S},\mathbf{CB})$ being a non-convex set in general.
	}
	
\subsection{Quadratic invariance for convexity given an arbitrary information structure}	
	If the set $h(\mathbfcal{S},\mathbf{CB})$ is convex, then the resulting constraints on $\mathbf{Q}$ are convex and Problem~2 is a convex program equivalent to Problem~1. Hence, the question arises as to when $h(\mathbfcal{S},\mathbf{CB})$ is a convex set.
%	 \begin{remark}
%\label{re:time_varying} 
% The class of general time-varying information structures as defined in (\ref{eq:information_structures_general}), (\ref{eq:information_structures_general2}) is strictly larger than the class of information structures based on fixed sparsity patterns and  fixed communication topologies considered in \cite{lamperski2015optimal,rotkowitz2005tractable,rotkowitz2010convexity}. Indeed, time-varying sparsity patterns and time-varying communication topologies can also be encoded in the framework we propose. This can be done by accordingly defining the sets $\mathcal{I}_k^i$ in (\ref{eq:information_structures_general}), (\ref{eq:information_structures_general2}).\\
% \end{remark}
%	 
%\subsection{Quadratic invariance for time-varying information structures}
%\label{sub:QIforconvexity}	

%The following Lemma~is adaptation of  \cite[Theorem~5]{QIconvexity} to finite dimensional linear operators.
\begin{lemma}
\label{le:2}
The set $h(\mathbfcal{S},\mathbf{CB})$ is convex if and only if it is equal to $\mathbfcal{S}$.
%Let $\mathbf{T}\in \mathbb{R}^{mN \times n(N+1)}$ be a convex set containing the origin and $h(\mathbf{S},\mathbf{B})=\mathbf{T}$. Then $\mathbf{T}=\mathbf{S}$.
\end{lemma}

Lemma~\ref{le:2} can be proved similar to  \cite[Theorem~5]{QIconvexity} by using finite dimensional operators, subspaces and convex sets whereas \cite{QIconvexity} uses Banach spaces over $\mathbb{R}$,  double-cones and star sets.
Lemma~\ref{le:2} motivates looking for conditions so that $h(\mathbfcal{S},\mathbf{CB})=\mathbfcal{S}$.  Such conditions are based on adapting the notion of quadratic invariance \cite{rotkowitz2006characterization} to the considered setting.
%Function $h(\cdot)$ can be verified to be an involution of the first argument, that is $h(h(\mathbf{X},\mathbf{Y}),\mathbf{Y})=\mathbf{X}$. 
\begin{definition}
\label{de:QI}
    The set $\mathbfcal{S}=\text{Sparse}(\mathbf{S})$ is quadratically invariant (QI) with respect to $\mathbf{CB}$ if and only if
    \begin{equation}
    \label{eq:QI_infinite}
    \mathbf{LCBL'} \in \mathbfcal{S}, \quad \forall \mathbf{L,L'} \in \mathbfcal{S}\,.
    \end{equation}
\end{definition}
Based on the same proof technique as per  \cite[Theorem~14]{rotkowitz2006characterization} the following result holds.
\begin{lemma}
\label{le:1}
The sets $ h(\mathbfcal{S},\mathbf{CB})$ and $\mathbfcal{S}$ are equivalent if and only if $\mathbfcal{S}$ is QI with respect to $\mathbf{CB}$.
%\begin{equation*}
%\mathbf{S}\text{ \emph{is QI w.r.t.} }\mathbf{B}\iff h(\mathbf{S},\mathbf{B}) = \mathbf{S}\,.
%\end{equation*}
\end{lemma}

One of the key insights into proving Lemma~\ref{le:1} is the power series expansion of the closed-loop map $h(\mathbf{X,Y})$. Notice that in the infinite dimensional setting of \cite{rotkowitz2006characterization} additional assumptions for plant and controller are needed to ensure  $\|\mathbf{XY}\|<1$ and convergence of the power series \cite[Section III-A]{rotkowitz2006characterization}, whereas in the finite dimensional setting considered here this condition is satisfied by construction.   Indeed, the diagonal of matrix $\mathbf{LCB}$ is null for any dynamical system and affine controller. The above discussion leads to the following result. %sufficient and necessary conditions for convexity of robust control problems g iven an arbitrary information structure.%, when using a disturbance-feedback parametrization.
\begin{proposition}
\label{pr:QI_general}
 Problem~2 is a convex program equivalent to Problem~1 if and only if $\mathbfcal{S}$ is QI with respect to $\mathbf{CB}$.%
%$\mathbfcal{S}$ is QI with respect to $\mathbf{CB}$.
%The set $\Pi_N^{\text{df,S}}(x_0)$ is convex if and only if $\mathbfcal{S}$ is QI with respect to $\mathbf{CB}$. In this case, the nonlinear sparsity constraint $\mathbf{Q}(\mathbf{CBQ}+I_{p(N+1)})^{-1} \in \mathbfcal{S}$ can be equivalently expressed as $\mathbf{Q} \in \mathbfcal{S}$.
\end{proposition}
\begin{proof}
\emph{
Problem~$1$ and Problem~$2$ are equivalent upon parametrizing the control input as an output-feedback policy (\ref{eq:affine_feedback}) or a disturbance-feedback policy (\ref{eq:dist_feedback}) respectively.  By combining Lemma~\ref{le:1} and Lemma~\ref{le:2} above we have that $h(\mathbfcal{S},\mathbf{CB})$ is convex if and only if $\mathbfcal{S}$ is QI with respect to $\mathbf{CB}$.  \QEDA}
%By \cite[Theorem~26]{rotkowitz2006characterization}, QI is equivalent  to
%\begin{align*}
%&\mathbf{S}(k,i)\mathbf{\Delta}(i,j)\mathbf{S}(j,l)(1-\mathbf{S}(k,l))=0,\\
%& \forall k,j \in \mathbb{N}_{[0,mN]},~\forall i,l \in \mathbb{N}_{[0,pN]}\,,
%\end{align*}
%which can be tested in the polynomial time $O(m^2p^2N^4)$. It is easy to show that the above is equivalent to
%\begin{equation*}
%\mathbf{S \Delta S}\leq \mathbf{S}\,,
%\end{equation*}
%which can be tested in the polynomial time $O(mp^2N^3)$ or $O(m^2pN^3)$ by exploiting the distributive property of the matrix product. 
%\QEDA
\end{proof}
When $h(\mathbfcal{S},\mathbf{CB})$ is convex, it is equal to $\mathbfcal{S}$ by Lemma~\ref{le:1}. Hence, given QI one can simply substitute (\ref{eq:sparsity_nonlinear}) with $\mathbf{Q} \in \mathbfcal{S}$ and obtain a finite dimensional convex program, efficiently solvable with standard convex optimization techniques. The optimal solution $(\mathbf{Q}^\star,\mathbf{v}^\star)$ is translated to the globally optimal affine output-feedback solution of Problem~1  $(\mathbf{L}^\star,\mathbf{g}^\star)$ through the mapping (\ref{eq:map_Q_to_L}).

Notice that QI as per (\ref{eq:QI_infinite}) cannot be tested in practice, as the subspace $\mathbfcal{S}$ contains infinite elements. Our main result is to translate convexity conditions into finitely many inequalities over binary matrices for any arbitrary information structure in finite-horizon.
\begin{theorem}
\label{th:characterization_general}
Let $\Delta_g=\text{Struct}(CA^gB)$. Then, $\mathbfcal{S}$ is QI with respect to $\mathbf{CB}$ if and only if
\begin{equation}
\label{eq:binary_small}
S_{k,h}\Delta_g S_{h-g-1,j} \leq S_{k,j}\,, 
\end{equation}
for all $k \in \mathbb{N}_{[1,N-1]}$, $j \in \mathbb{N}_{[0,k-1]}$, $h \in \mathbb{N}_{[j+1,k]}$ and $g \in \mathbb{N}_{[0,h-j-1]}$.

\end{theorem}
\begin{proof}
\emph{
Let $\mathbf{\Delta}=\text{Struct}(\mathbf{CB})$. By \cite[Theorem~26]{rotkowitz2006characterization}, QI is equivalent  to
\begin{align*}
&\mathbf{S}(k,i)\mathbf{\Delta}(i,j)\mathbf{S}(j,l)(1-\mathbf{S}(k,l))=0,\\
& \forall k,j \in \mathbb{N}_{[0,mN]},~\forall i,l \in \mathbb{N}_{[0,pN]}\,.
\end{align*}
%which can be tested in the polynomial time $O(m^2p^2N^4)$. 
By using the definition of the matrix product, the above is equivalent to
\begin{equation}
\label{eq:binary_big}
\mathbf{S \Delta S}\leq \mathbf{S}\,.
\end{equation}
It remains to decompose (\ref{eq:binary_big}) into (\ref{eq:binary_small}). Let us define $\bm{\varphi}=\mathbf{\Delta S}$ and $\mathbf{\Phi}=\mathbf{S}\bm{\varphi}=\mathbf{S\Delta S}$. Let $\varphi_{k,j}$ and $\Phi_{k,j}$ denote their $p \times p$ and $m \times p$ block sub-matrices respectively, located at block-row $k$ and block-column $j$.  Observe that
\begin{equation*}
\varphi_{h,j}=\sum_{g=0}^{h-j-1}\Delta_g S_{h -g -1, j}\,,
\end{equation*}
for $h \in \mathbb{N}_{[1,N]}$, $j\in \mathbb{N}_{[0,k-1]}$ and $\varphi_{h,j}=0_{p \times p}$ otherwise. Pre-multiplying by $\mathbf{S}$, we obtain
\begin{align}
\Phi_{k,j}&=\sum_{h=j+1}^{k}S_{k,h}\varphi_{h,j}\nonumber\\ 
&=\sum_{h=j+1}^{k}S_{k,h}\sum_{g=0}^{h-j-1}\Delta_g S_{h -g -1, j}\,, \label{eq:binary_small_sums}
\end{align} 
for $k \in \mathbb{N}_{[1,N-1]}$, $j \in \mathbb{N}_{[0,k-1]}$ and $\Phi_{k,j}=0_{m \times p}$ otherwise. Then, (\ref{eq:binary_big}) is equivalent to $\Phi_{k,j} \leq S_{k,j}$ for every $k \in \mathbb{N}_{[1,N-1]}$ and $j \in \mathbb{N}_{[0,k-1]}$. Given (\ref{eq:binary_small_sums}),  $\Phi_{k,j} \leq S_{k,j}$  if and only if each addend in its expression is less or equal than $S_{k,j}$. This concludes the proof. \QEDA}
%which can be tested in the polynomial time $O(mp^2N^3)$ or $O(m^2pN^3)$ by exploiting the distributive property of the matrix product. 
\end{proof}

Our  test (\ref{eq:binary_small}) shows that convexity  can be certified given any arbitrary information structure in finite-horizon, by verifying a finite number of inequalities over binary matrices. % In particular,  including intricate time-varying cases such as those considered in \cite{matni2014optimal}, \cite{abara2018quadratic} and \cite{ouyang2017stochastic} which can be encoded as exploined in Section~\ref{se:preliminaries}. 
Instead, the work \cite{rotkowitz2006characterization} introduced a binary test  which is  applicable to the special case of time-invariant sparsity constraints. Algebraic conditions on the propagation and transmission delays for the case of combining time-invariant sparsity constraints and  delayed communication were proposed in \cite{rotkowitz2010convexity}. However, finitely many inequalities equivalent to QI were not characterized for more general information structures, such as the intricate time-varying cases studied in \cite{matni2014optimal} and \cite{ouyang2017stochastic}.% Hence, we unify  previous  tests of QI and significantly generalize the information structures for which a finite number of inequalities is enough to certify convexity.% without the need of introducing new notions about the specific mechanisms of a certain information structure,  and by only inspecting appropriate binary matrices.

\begin{example}
\emph{
Here, we use our proposed binary test (\ref{eq:binary_small}) to identify a complex QI information structure, that involves a combination of time-varying sensing, communication and the possibility for controllers to forget the outputs they have either measured or received. Consider the following three-dimensional system with two scalar controllers, whose dynamics (\ref{eq:system}) are defined by the matrices
\begin{equation*}
A=\begin{bmatrix}
0&0&1\\-2&0&0\\0&3&0
\end{bmatrix}\,, \quad B= \begin{bmatrix}
1&0\\1&0\\0&1
\end{bmatrix}\,, \quad C=I_3\,,
\end{equation*}
and matrices $D$ and $H$ are chosen arbitrarily. We consider a prediction horizon of $N=3$. The information structure is defined as follows. At time $k=0$, $u^2$ measures $y_0^1$. At time $k=1$, both $u^1$ and $u^2$ measure $y_1^1$. Additionally, $u^1$ receives information about $y_0^3$ and $u^2$ receives information about $y_0^2$ and $y_0^3$. However, $u^2$ forgets $y_0^1$. At time $k=2$, $u^1$ measures $y_2^2$ and receives information about $y_1^3$, $y_0^1$ and $y_0^2$, while $u^2$  does not perform any measurement and receives information about $y_1^2$. However, $u^2$ forgets $y_0^2$. This rather complicated information structure is conveniently encoded by the matrices
\begin{alignat*}{3}
&S_{0,0}=\begin{bmatrix}
0&0&0\\1&0&0
\end{bmatrix}\,,&&\quad S_{1,0}=\begin{bmatrix}
0&0&1\\0&1&1
\end{bmatrix}\,,\\
&S_{1,1}=\begin{bmatrix}
1&0&0\\1&0&0
\end{bmatrix}\,,&&\quad S_{2,0}=\begin{bmatrix}
1&1&1\\0&0&1
\end{bmatrix}\,,\\
&S_{2,1}=\begin{bmatrix}
1&0&1\\1&1&0
\end{bmatrix}\,,&& \quad S_{2,2}=\begin{bmatrix}
0&1&0\\0&0&0
\end{bmatrix}\,.
\end{alignat*}
Our goal is to certify whether convex design of optimal robust controllers complying with such an information structure is possible. %allows for convex design of robust distributed controllers by solving Problem~2 with standard convex optimization techniques. 
 % that if a controller $u^i$ receives $y^j_{k-k'}$ after a  delay $k'$, it also knows $y^j_{k-l}$ for each $l \in \mathbb{N}_{[0,k'-1]}$. In our case, however, at time $k=2$ the controller $u^2$ knows $y_1^2$ but not $y_0^2$.  The case of intermittent observations recently considered in \cite{abara2018quadratic} works under the assumption that controllers never forget the information they either observed or received over time (see equation (4) in \cite{abara2018quadratic}).
By letting $\Delta_0=\text{Struct}(CB)=B$ and $\Delta_1=\text{Struct}(CAB)$, we  apply the generalized test for convexity developed in Theorem~\ref{th:characterization_general}, which amounts to verifying the following set of inequalities:
\begin{alignat*}{3}
&S_{1,1}\Delta_0 S_{0,0}\leq S_{1,0}\,,&&\quad S_{2,1}\Delta_0 S_{0,0} \leq S_{2,0}\,,\\
&S_{2,2} \Delta_0 S_{1,0} \leq S_{2,0}\,,&& \quad S_{2,2} \Delta_1 S_{0,0} \leq S_{2,0}\,,\\
&S_{2,2} \Delta_0 S_{1,1} \leq S_{2,1}\,.
\end{alignat*}
Since each of the inequalities above holds, we conclude that the considered information structure allows for convex design of optimal robust distributed controllers. }

\emph{We remark that the convexity tests proposed in \cite{rotkowitz2006characterization, rotkowitz2010convexity} cannot be applied in this case, as the information structure cannot be described by a time-invariant sparsity pattern with delays. The cases of time-varying delays considered in \cite{matni2014optimal} and that of intermittent observations recently considered in \cite{abara2018quadratic} assume that controllers never forget the information they have either observed or received over time. Hence, to the best of the authors' knowledge, the QI information structure considered here is not captured by other work. This example we consider reveals that convexity can be preserved even if controllers forget information in a finite-horizon.}
\end{example}

\emph{Interpretation of Theorem~\ref{th:characterization_general}}. The binary test (\ref{eq:binary_small}) naturally allows for a generalized interpretation of convexity in terms of the information which must be available to controllers as follows.

``Whenever  controllers  at time $k$  know  some information about the outputs $y_h$ at time $h\leq k$, they must also know the output information which was available  to those controllers  whose decisions influenced these values of $y_h$ through the evolution of the dynamics.''

The above interpretation of (\ref{eq:binary_small}) is a direct consequence of the fact that $S_{k,j}$ encodes what the controllers at time $k$ know about the outputs at a past time $j$ and that $\Delta_g$ encodes which outputs are affected after $g$ time steps by the decisions of controllers. We remark that the above is consistent with the idea of \emph{signaling} \cite{mahajan2012information}, that is  the fact that controllers might need to infer information they don't have by inverting the system dynamics, thus compromising convexity of the corresponding optimization problem. Each binary inequality in (\ref{eq:binary_small}) can thus be thought of as eliminating signaling between controllers at two specific different times.

	\section{The Case of Fixed Sensing and Communication Topologies}
	\label{sec:dyn_complete}

In this section, we specialize  the argument of Theorem~\ref{th:characterization_general} to specific information structures that are commonly found in practice. Our goal is to show  that the convexity results of \cite{rotkowitz2010convexity} can be alternatively proved as a direct corollary of Theorem~\ref{th:characterization_general}, thus allowing for additional insight on the minimal number of inequalities to be tested for convexity. For simplicity, we start with the case where a communication network is not available.
%Second, we explain in a purely algebraic way the reason why a communication network must be available to achieve convexity when treating systems whose dynamics evolve according to a strongly connected topology.  For simplicity, we start with the case where a communication network is not available.

\subsection{Sensor information structures}
%Here, we consider the case of time-invariant sparsity constraints, where controllers can either measure a specific output at all times, or they can never receive any information about such output. This  information structure,  also considered in the classical works \cite{rotkowitz2006characterization,rotkowitz2010convexity} and denominated as ``sparsity constraints'', was  addressed in \cite{furieri2017robust} within the context of robust distributed control.  %This scenario arises when controllers can directly measure some outputs through a sensor network, but they can never communicate with each other. 

We define a \emph{sensing topology} that indicates which outputs are directly measured by  which controllers. The sensing topology can be conveniently described  by a binary matrix $S \in \{0,1\}^{m \times p}$ such that $S(i,j)=1$ if and only if controller $i$  can measure $y^j_k$ at every time instant $k$. As outlined in Section~\ref{se:preliminaries}, we can then set $S_{k,j}=S$ for all $k,j$ to encode the corresponding information structure. Here, we show that \cite[Theorem~3]{furieri2017robust} about conditions on the system reachability matrix and the sensing topology for convexity can be  alternatively proved as an immediate corollary of Theorem~\ref{th:characterization_general}. First, we state the following lemma.
	\begin{lemma}
	\label{le:binaries}
	Let $X$ and $Y$ be binary matrices of dimensions $m \times n$. Let $Z,Z'$ be binary matrices of dimensions $p \times m$ and $n \times p$ respectively. Then, $X\leq Y$ implies that $ZX \leq ZY$ and $XZ'\leq YZ'$.
		\end{lemma}
		\begin{proof}
		\emph{Let us suppose $X \leq Y$.  We want to show that $ZX(i,j)=1$ implies that  $ZY(i,j)=1$ for all indices $i,j$. If $ZX(i,j)=1$, then there exists a $k$ such that $Z(i,k)=X(k,j)=1$. By hypothesis, also $Y(k,j)=1$. Then $ZY(i,j)=1$. $XZ' \leq YZ'$ is proved analogously.  \QEDA }
		\end{proof}

%The conditions  for convexity (\ref{eq:binary_small}) simplify significantly as follows. We show that \cite
\begin{corollary}[of Theorem~\ref{th:characterization_general}]
Let $S \in \{0,1\}^{m \times p}$ denote the sensing topology matrix and let $S_{k,j}=S$ for every $k\in \mathbb{N}_{[0,N-1]}$ and $j \in \mathbb{N}_{[0,k]}$. Let $\Delta_g=\text{Struct}(CA^gB)$.  Then, Problem~2 is a convex program equivalent to Problem~1 if and only if %Then, $\mathbfcal{S}$ is QI with respect to $\mathbf{CB}$ if and only if
\begin{equation}
\label{eq:sensing_topology_convexity}
S\Delta_g S \leq S, \quad \forall g \in \mathbb{N}_{[0,n-1]}\,.
\end{equation}
\end{corollary}
\begin{proof}
\emph{By Proposition~\ref{pr:QI_general} and Theorem~\ref{th:characterization_general}, Problem~2 is a convex program equivalent to Problem~1 if and only if (\ref{eq:binary_small}) holds. Clearly, (\ref{eq:binary_small}) reduces to (\ref{eq:sensing_topology_convexity}). It remains to be proven that it is sufficient to consider $g \in \mathbb{N}_{[0,n-1]}$. Suppose that $v\geq n$. Then, by Cayley-Hamilton we have that $CA^vB=C\sum_{i=0}^{n-1}\lambda_i A^i B$ for some coefficients $\lambda_i \in \mathbb{R}$. It follows that $\Delta_v \leq \sum_{i=0}^{n-1}\Delta_i$, because the coefficients $\lambda_i$ might lead to some entries of $\Delta_v$ being null despite the same entry being $1$ in $\Delta_i$ for some $i \in \mathbb{N}_{[0,n-1]}$. But then, if (\ref{eq:sensing_topology_convexity}) holds, we have by Lemma~\ref{le:binaries} that
\begin{equation*}
S\Delta_v S \leq S \sum_{i=0} ^{n-1}\Delta_i S \leq S\,.
\end{equation*}
In other words, $S\Delta_g S \leq S$ for all $g \in \mathbb{N}_{[0,n-1]}$ implies  $S\Delta_v S \leq S$ for all $v \in \mathbb{N}$. \QEDA}
\end{proof}

By \cite[Theorem~26]{rotkowitz2006characterization}, the conditions (\ref{eq:sensing_topology_convexity}) imply that $\text{Sparse}(S)$ is QI with respect to $CA^gB$ for every $g \in \mathbb{N}_{[0,n-1]}$, thus extending \cite[Theorem~3]{furieri2017robust} to the output-feedback case. We remark that while our binary test  for convexity (\ref{eq:binary_small}) involves verifying a number of inequalities that is polynomial in $N$ in general, the corollary above shows that only $n$ inequalities are necessary for the class of fixed sensor information structures.

An inherent limitation of time-invariant sensor information structures appears evident from (\ref{eq:sensing_topology_convexity}). Indeed, when the system matrix $A$ describes a strongly connected topology, we have that $CA^{n-1}B$ is a dense matrix in general. It follows that the product $S\Delta_{n-1} S$ is dense whenever $S$ contains at least a $1$ in each row and column, and hence that $S$ must also be dense in general to satisfy (\ref{eq:sensing_topology_convexity}). This makes convex design of distributed controllers impossible given a fixed sensor information structure. We conclude that a communication network that propagates the output measurements across controllers  must be available in order to restore a convex optimization problem for strongly connected systems. In the next section, we exploit Theorem~\ref{th:characterization_general} to investigate the value of communication in allowing for convex design of distributed controllers.

\subsection{Sensor and communication information structures}
Suppose that a time-invariant sensing topology $S$ is defined and that a communication network is also available. Controllers that can exchange information are encoded by a corresponding \emph{communication topology}, where information sharing between controller $i$ and controller $j$ is unidirectional in general. It is convenient to encode the communication topology in a binary matrix $Z \in \{0,1\}^{m \times m}$  such that $Z(i,j)=1$ if and only if controller $i$ receives information from controller $j$ at each time $k$. In many applications, controllers are endowed with memory and thus ``receive'' information from themselves. For this reason we pose $Z(i,i)=1$ for all $i$.

Communication could be exploited by simply letting controllers exchange their direct measurements, obtained according to the sensing topology. In this case, we would have that $S_{k,j}=ZS$ for every $k,j$. However, similar to the case of sensor information structures, $Z$ would need to be dense in general for convexity of Problem~2 given a strongly connected system.  To go beyond this limitation, we let controllers memorize both their direct sensor measurements and the information they receive through communication, and then propagate everything they have stored in their memory to their neighbours in the communication topology, who receive this information one time step later.

 In other words, at every time instant $k \in \mathbb{N}_{[0,N -1]}$, controller $u^a$ knows $y^l$ at time $k-r$ if and only if there is a controller $u^b$ which can directly measure $y^l$ according to $S$, and there is a path of length $r$ in the communication topology $Z$ from $u^b$ to $u^a$. In the next lemma, we encode this requirement as a particular choice for the sparsity subspace $\mathbfcal{S}$. First, we recall the definition of diameter of graph. %in terms of sensing and the communication topologies denoted as $S$ and $Z$ respectively. First, we recall the definition of diameter of graph.
	\begin{definition}[Diameter of a directed graph]
	\label{de:diameter}
	Consider a directed graph $\mathcal{G}$. We define the diameter of $\mathcal{G}$ to be the largest number of nodes contained in any of its paths when backtracking, detouring and looping paths are excluded from consideration and we denote it as $\mathcal{D}(X)$, where $X$ is the binary matrix describing which nodes of $\mathcal{G}$ are connected by an edge. %In other words, it is the length of the \emph{longest shortest path} between nodes in the graph. Whenever $X$ is a binary matrix, $\mathcal{D}(X)$ indicates the diameter of the graph having $X$ as its adjacency matrix.
%More formally,
%	%Suppose the graph is connected, that is there exists a node starting from which all the other nodes in the graph can be reached through a path. 
%	%Then, we define the diameter of $\mathcal{G}$ as
%	\begin{alignat*}{3}
%	&\mathcal{D}(R)=
%	&&\min_{r \in \mathbb{N}_{[0,|\mathcal{V}|]}} r\\
%	&\text{s.t. }&&\text{Struct}(R^r)=\text{Struct}(R^{r+1})\,.
%\end{alignat*}
	%that is the length of the longest path in the graph which is also in the set of shortest paths between nodes. 	In the paper we denote $\mathcal{D}(A)=\mathcal{D}(\mathcal{G}_A)$ and $\mathcal{D}(Z)=\mathcal{D}(\mathcal{G}_Z)$ to ease notation.
	\end{definition}
%	The sparsity pattern of the blocks in $\mathbf{S}_Z$ stops changing when the powers of $Z$ reach the value $\mathcal{D}(Z)$. To show this fact we first  introduce a lemma whose proof is provided in Appendix \ref{app:growing_powers}. %The needed operators over binary matrices are defined in Appendix \ref{app:notation}.
%		\begin{lemma}
%	\label{le:growing_powers}
%Let $A^{\text{bin}}\in \{0,1\}^{a \times a}$. Assume that $A^{\text{bin}}$ has no null diagonal  entries. Let $X,Y \in \{0,1\}^{a \times b}$ and $S \in \{0,1\}^{b \times a}$. Then
%	\begin{equation}
%		\label{eq:le1}
%	{A^{\text{bin}}}^{r}X\leq {A^{\text{bin}}}^{r+1}X, ~ \forall r \in \mathbb{N}_{[0,\infty)}\,.
%	\end{equation}
%	Additionally,
%	\begin{equation}
%	\label{eq:le2}
%	{A^{\text{bin}}}^sX={A^{\text{bin}}}^{s+1}X,~\forall s \in \mathbb{N}_{[\mathcal{D}(A^{\text{bin}}),\infty)}\,.
%	\end{equation}
%	Lastly,
%	\begin{equation}
%	\label{eq:le3}
%	\begin{aligned}
%	&X \leq Y \implies SX \leq SY~\&~ XS \leq YS \,.
%	%&	X \leq Y \implies XS \leq YS\,.
%		\end{aligned}
%	\end{equation}
%	\end{lemma}

%\subsection{Sparsity constraints for P-ISCT information structures}
	%The following theorem establishes sparsity constraints equivalent to a given P-ISCT information structure.
	\begin{lemma}
	\label{le:sparsity constraints}
	Let $S \in \{0,1\}^{m \times p}$ be the sensing topology matrix  and $Z \in \{0,1\}^{m \times m}$ be the communication topology matrix. The information structure defined above can be encoded by selecting the matrices $S_{k,j}$ in (\ref{eq:affine_feedback_constraints}) as % Suppose that  at every time instant $k \in \mathbb{N}_{[0,N -1]}$, controller $u^a$ must know $y^l$ at time $k-r$ if and only if there is a controller $u^b$ which can directly measure $y^l$ according to $S$, and there is a path of length $r$ in the communication topology $Z$ from $u^b$ to $u^a$. %Also consider the case that $\mathcal{D}(Z)\leq N-1$. 
%	The corresponding information structure is encoded by selecting the matrices $S_{k,j}$ in (\ref{eq:affine_feedback_constraints}) as
	\begin{equation}
	\label{eq:characterize_propagating_communication}
	S_{k,j}=Z^{\min (\mathcal{D}(Z),k-j)}S~\,.
%	\begin{equation*}
%	\label{eq:SA}
%	\mathbf{L} \in \mathbfcal{S}\subseteq \mathbb{R}^{m(N+1) \times p(N+1)}\,,
%	
\end{equation}

	\end{lemma}
	\begin{proof}
	\emph{
It is well-known that $Z^r$ encodes the paths of length $r$ in the communication topology \cite{godsil2013algebraic}. Then, by the definition of the information structure as above,  we have that $u^a$ knows $y^l_{k-r}$ if and only if there exists a controller $u^b$ such that $S(b,l)=1$ and $Z^r(a,b)=1$, or equivalently  $(Z^rS)(a,l)=1$. Now, for any $k \in \mathbb{N}_{[0,N-1]}$ and $r \in \mathbb{N}_{[0,k]}$, by definition  the matrix $S_{k,k-r}$ must encode the output measurements known to controllers at time $k$  about the outputs at time $k-r$. Hence, $S_{k,k-r}=Z^rS$ or equivalently $S_{k,j}=Z^{k-j}S$ for every $j \in \mathbb{N}_{[0,k]}$. Now, if $\mathcal{D}(Z)\leq k-j$, then $Z^{k-j}S=Z^{\mathcal{D}(Z)}S$ because $Z \geq I_n$ and by the definition of diameter of a graph. Equation (\ref{eq:characterize_propagating_communication}) follows. \QEDA  }%The result is proved by repeating the reasoning for all block sub-matrices $[\mathbf{S}_Z]_{i,j}$ such that $j\leq i \leq N$, using the observation in the previous sentence to discern between (\ref{eq:S_DC_binary_N}) and (\ref{eq:S_DC_binary_D}) and setting all the other blocks to $0_{m \times p}$ for causality. \QEDA
\end{proof}

%%\begin{remark}
%Consider the graph having $Z$ as adjacency matrix, and suppose $\mathcal{D}(Z)<N$. The powers of $Z$ according to the binary product $\boxdot$ stop growing for exponents greater than $\mathcal{D}(Z)$ due to Lemma~\ref{le:growing_powers}. Accordingly, the blocks of $\mathbf{S}_Z$  stop growing as well.  This means that all delayed output information can be gathered within $\mathcal{D}(Z)$ time steps, independent of the control horizon $N$.% highlighting the importance of a proper communication topology design. %All the block sub-matrices in $\mathbf{S}_Z$ of the form $Z^rS,~r \geq \mathcal{D}(Z)$ can be substituted by $Z^{\mathcal{D}(Z)}S$.
%\end{remark}

	%Applying the QI argument shown in \ref{sub:QIforconvexity}, 
% Equation (\ref{eq:characterize_propagating_communication}) specifies that all the delayed output information can be gathered within $\mathcal{D}(Z)$ time steps, independent of the control horizon $N$.	
 Conditions for QI given a similar information structure were derived in \cite{rotkowitz2010convexity}, based on the notions of propagation and transmission delays between subsystems. Here, we show that a test of convexity in finite-horizon can be alternatively derived as a direct corollary of our Theorem~\ref{th:characterization_general}. We report the proof in Appendix~\ref{app:SZ_QI}.%Next, we investigate how this class of information structures can enable for convex design of distributed controllers in the presence of strongly connected systems, by specializing the convexity test developed in (\ref{eq:binary_small}). We report the proof of the following corollary of Theorem~\ref{th:characterization_general} in the Appendix.
	%\subsection{Conditions for convexity given P-ISCT information structures}
	%\label{sub:info_sharing_scheme_QI}
	%In order to guarantee convexity of the set $\Pi_N^{\text{df,S}}(x_0)$ defined in (\ref{eq:Pi_df_tilda}) when implementing a $(S,Z)$-information structure, we know by theorem \ref{th:QI_general} that the necessary and sufficient condition is that the corresponding sparsity subspace $\mathbfcal{S}$ is QI with respect to $\mathbf{CB}$. 
%Consider Problem 1 given a P-ISCT information structure. We showed in Section \ref{sec:Robustfinitehorizondistributed} that Problem 2 can be cast as a convex program equivalent to Problem 1 if and only if $\mathbfcal{S}$ is QI with respect to $\mathbf{CB}$. More specifically, it is required that $\mathbf{\Phi}=\mathbf{LCBL'} \in \mathbfcal{S}$ for any choice of  $\mathbf{L,L'}\in \mathbfcal{S}$. Verifying the above condition is not straightforward. Hence, we derive a set of equivalent conditions in terms of inequalities involving system matrices and the sensing and communication topology matrices. Based on such  conditions a graph theoretic test for QI is derived in Section \ref{sec:info_sharing_graphic_interpretation}, which can give valuable visual insight when designing the network. The next theorem is the main result of this section. Its proof is reported in Appendix \ref{app:SZ_QI}. %Before proceeding, recall the operators over binary matrices introduced in Appendix \ref{app:notation}.\\
\begin{corollary}[of Theorem~\ref{th:characterization_general}]
\label{co:propagating}
Let $S \in \{0,1\}^{m \times p}$ denote the sensing topology matrix  and $Z \in \{0,1\}^{m \times p}$ denote the communication topology matrix. Let $S_{k,j}$ be chosen according to (\ref{eq:characterize_propagating_communication}). Let $\Delta_g=\text{Struct}(CA^gB)$.  Then, Problem~2 is a convex program equivalent to Problem~1 if and only if
	\begin{equation}
	\label{eq:convexity_propagating}
	\begin{aligned}
	& \qquad \qquad \qquad S\Delta_gZ^{r}S \leq Z^{g+r+1}S,\\
	&\forall g \in \mathbb{N}_{[0,n-1]},~\forall r \in \mathbb{N}_{[0, \mathcal{D}(Z)]}\text{ s.t. }g+r\leq N-2\,.
	%&\forall k \in \mathbb{N}_{[0,\bar{D}_r]}, r \in \mathbb{N}_{[0,\mathcal{D}(Z)]}\,
	\end{aligned}
	\end{equation}
\end{corollary}
%\begin{theorem}
%	\label{th:SA_QI}
%	Consider a dynamical system as in (\ref{eq:system}) and the corresponding matrices $A,B,C$.  Let $S$ and $Z$ be the sensing  and  communication topology matrices.
%	%and that the graph which has $Z$ as its adjacency matrix is connected. 
%	Let us compute matrix $\mathbf{S}_{Z}$ as in (\ref{eq:S_DC_binary_N}).
%	%Define  $\mathcal{A}=\text{Struct}(A)$,  $\mathcal{B}=\text{Struct}(B)$, $c=\text{Struct}(C)$. 
%	For every integer $k$ let us define
%\begin{equation}
%\label{eq:Delta}
%\Delta_k=\text{Struct}\left(CA^kB\right)\,.
%\end{equation}	
%	Then,  $\mathbfcal{S}$ is QI with respect to $\mathbf{CB}$ if and only if
%	\begin{equation}
%	\label{eq:condition_lemma1}
%	\begin{aligned}
%	&S\Delta_kZ^{r}S \leq Z^{k+r+1}S,\\
%	&\forall k \in \mathbb{N}_{[0,n-1]},~\forall r \in \mathbb{N}_{[0, \mathcal{D}(Z)]}\text{ s.t. }k+r\leq N-2\,.
%	%&\forall k \in \mathbb{N}_{[0,\bar{D}_r]}, r \in \mathbb{N}_{[0,\mathcal{D}(Z)]}\,
%	\end{aligned}
%	\end{equation}
%	\end{theorem}		
 Equation (\ref{eq:convexity_propagating}) shows in a purely algebraic way  the reason why communication enables convexity for strongly connected systems. Indeed, notice that the growth of the left-hand-side of (\ref{eq:convexity_propagating}) as $g$ grows can be accommodated by the growth of the right-hand-side thanks to $Z^g$. Instead, in the case of sensor information structures where $Z=I_m$, the right-hand side cannot grow, thus compromising convexity for strongly connected systems.  %We conclude that introducing communication with propagation is necessary in general to restore convex design of distributed controllers for strongly connected systems. %When instead $Z=I-m$ To understand this fact observe that the growing powers of $Z>I_m$  on the right-hand-side of (\ref{eq:condition_lemma1}) facilitate satisfaction of the inequality as $k$ grows. That is, the growth of the left-hand-side of (\ref{eq:condition_lemma1}) is accommodated by the growth of the right-hand-side of  (\ref{eq:condition_lemma1}). When instead $Z=I_m$  the left-hand-side grows with $k$ whereas the right-hand-side does not, thus making the conditions difficult to satisfy for sparse $S$'s.

We remark that Corollary~\ref{co:propagating}, while being in accordance with \cite[Theorem~2]{rotkowitz2010convexity}, offers new insight. First, it allows treating sensing and communication as two completely independent time-invariant topologies. Second, it identifies the minimal number of inequalities to be verified for convexity in terms of the system dimension $n$ and the communication network diameter $\mathcal{D}(Z)$, whereas  \cite{rotkowitz2010convexity} requires verifying one inequality for each pair of subsystems in terms of the propagation and transmission delays.

 \section{Conclusions}
We have unified and generalized the class of systems and information structures for which convexity of the robust distributed control problem in finite-horizon can be certified by finitely many algebraic conditions. In particular, we have shown that  the information structures separately treated  in the past works can all be described within the framework we suggest, and we derived a test for QI consisting of finitely many inequalities over certain binary matrices. Given the generality of the considered framework, we recovered previous results about specific classes of information structures as direct corollaries of our generalized test for convexity. In doing so, we provided new insight on the conditions for convexity given a fixed sensing and communication topology.

This work can be extended in several directions. First, it would be interesting to use the tools developed in this paper to address the case of time-varying system dynamics and study their interaction with event-dependent information structures such as those considered in \cite{abara2018quadratic}. Second, Example~1 suggests that forgetting mechanisms can also be included in the information structure without necessarily compromising convexity. Hence, it would be relevant to address the topic  of designing information structures which are robustly convex despite random packet dropouts and/or disorderings  \cite{wang2012wide} and despite faults in the internal memory of controllers.

\begin{appendices}
\allowdisplaybreaks
	\section{Additional Definitions}
	\label{app:notation}
	%\label{sub:math}
	
	We define the following matrices and vectors. 
		\begin{equation}
	\label{eq:L}
	%\medmuskip=0mu
%\thinmuskip=0mu
%\thickmuskip=0mu
%\nulldelimiterspace=-1pt
%\scriptspace=0pt
	\begin{aligned}
	&\mathbf{L}\hspace{-0.07cm}=\hspace{-0.07cm}
		\begingroup % keep the change local
\setlength\arraycolsep{1.4pt}
\begin{bmatrix}
	L_{0,0}&0_{m \times p}&\cdots&0_{m \times p}\\
	%L_{1,0}&L_{1,1}&\ddots&0_{m \times p}\\
	\vdots&\vdots&\ddots&\vdots\\
	L_{N-1,0}&\cdots&L_{N-1,N-1}&0_{m \times p}\\
	0_{m \times p}&\cdots&0_{m \times p}&0_{m \times p}
	\end{bmatrix}\hspace{-0.125cm},~\mathbf{g}\hspace{-0.08cm}=\hspace{-0.08cm}\hspace{-0.1cm}\begin{bmatrix}g_0\\\vdots\\g_{N-1}\\0_{m \times 1}\end{bmatrix}\endgroup. \hspace{-0.1cm}
	\end{aligned}
	\end{equation}

	The matrix blocks above are $L_{k,j}\in \mathbb{R}^{m \times p}$, $g_k \in \mathbb{R}^{m}$ as in (\ref{eq:affine_feedback_constraints}), and the $0_{m \times p}$ blocks enforce causality of the controller. We also define the following matrices, where $\otimes$ denotes the Kronecker product.
	\begin{align*}
	&\mathbf{A}=\begin{bmatrix}I_n&A^\mathsf{T}&\cdots&{A^N}^\mathsf{T}\end{bmatrix}^\mathsf{T}\in \mathbb{R}^{n(N+1)\times n}\,,\\
	&\mathbf{E}=\hspace{-0.1cm}
		\begingroup % keep the change local
\setlength\arraycolsep{3pt}	
	\begin{bmatrix}0_{n \times n}&0_{n \times n}&\cdots&0_{n \times n}&0_{n \times n}\\
	I_n&0_{n \times n}&\cdots&0_{n \times n}&0_{n \times n}\\
	A&I_n&\cdots&0_{n \times n}&0_{n \times n}\\
	\vdots&\vdots&\ddots&\vdots&\vdots\\
	A^{N-1}&A^{N-2}&\cdots&I_n&0_{n \times n}
	\end{bmatrix}\endgroup\hspace{-0.1cm}\in  \mathbb{R}^{n(N+1)\times n(N+1)},\\
	&\mathbf{B}=\mathbf{E}(I_{N+1} \otimes B)\,, \quad  	\mathbf{E}_D=\mathbf{E}(I_{N+1} \otimes D)\,,\\
	%&\mathbf{\tilde{E}}_D=\begin{bmatrix}\mathbf{A}&\mathbf{E}_D\end{bmatrix}\in \mathbb{R}^{n(N+1)\times n(N+1)}\,,\\
	&\mathbf{C}=I_{N+1} \otimes C\,, \quad \mathbf{H}=I_{N+1} \otimes H \,,\\
	%&\mathbf{H}=I_{N} \otimes H \in \mathbb{R}^{p(N+1)\times nN}\\
	%&\mathbf{H}=\begin{bmatrix}0_{pN \times 1}&I_{N} \otimes H\\0_{n\times 1}&0_{n \times nN}\end{bmatrix} \in \mathbb{R}^{pN+n \times n(N+1)}\,, 
	&\mathbf{U}=\begin{bmatrix}
	I_N\otimes U & 0_{Ns \times n}\\0_{r \times nN}&R
	\end{bmatrix}\,, \quad \mathbf{V}=\begin{bmatrix}
	I_{N}\otimes V&0_{Ns \times m} \\ 
	0_{r \times mN}&0_{r \times m}
	\end{bmatrix}\,,\\
	&F=\mathbf{UB+V}\,, \quad G=\mathbf{UE}_D\,,\quad c=\begin{bmatrix} \mathbf{1}_N\otimes b\\z\end{bmatrix}-\mathbf{UA}x_0 \,.\\
	%&H=\mathbf{CA}\in \mathbb{R}^{(Ns+r)\times n}\,, \\
	\end{align*}

\section{Proof of Corollary~\ref{co:propagating}}
	\label{app:SZ_QI}

%By comparison of the corresponding blocks,  notice that the condition $\mathbf{\Phi} \in \mathbfcal{S}$  is verified if and only if $\text{Struct}(\Phi_{i,i-j})\leq Z^jS$ for all $i \in \mathbb{N}_{[0,N-1]}$ and $j \in \mathbb{N}_{[1,i]}$.  Hence, we will prove the following:
%	\begin{equation*}
%	\label{eq:th_QI_requiredCondition}
%	\begin{aligned}
%	&\text{Struct}(\Phi_{i,i-j})\leq Z^jS,\quad \forall i \in \mathbb{N}_{[0,N-1]},\forall j \in \mathbb{N}_{[1,i]},\\
%	&\qquad \qquad \qquad \qquad \qquad ~ \forall~{\mathbf{L,L'} \in \mathbfcal{S}}\,.\\
%	&\qquad \qquad \qquad \qquad \quad \iff\\
%	&\qquad \qquad \qquad \qquad ~  \text{(\ref{eq:condition_lemma1}) holds.}
%	\end{aligned}
%	\end{equation*}	
	By Proposition~\ref{pr:QI_general}, Problem~2 is a convex program equivalent to Problem~1 if and only if $\mathbfcal{S}$ is QI with respect to $\mathbf{CB}$. We thus prove necessity and sufficiency of (\ref{eq:convexity_propagating}) for QI.	
	\subsection{Proof of sufficiency} 
First, observe that if (\ref{eq:convexity_propagating}) holds for all $g \in \mathbb{N}_{[0,n-1]}$, then (\ref{eq:convexity_propagating}) holds for any $v\geq n$. Indeed
\begin{alignat*}{3}
S\Delta_vZ^{r}S &\leq S\sum_{i=0}^{n-1}\Delta_i Z^r S&& \quad (\text{Cayley-Hamilton and Lemma~\ref{le:binaries}})\\
&\leq Z^{n+r}S&& \quad (\text{by (\ref{eq:convexity_propagating}) and }Z\geq I_m)\\
& \leq Z^{v+r+1}S&& \quad (\text{by }Z\geq I_m)\,.
\end{alignat*}
Also notice that it is redundant to consider $r> \mathcal{D}(Z)$ by the definition of diameter and $Z \geq I_m$, and that we can restrict to $g+r\leq N-2$ due to the horizon length. Now, let us pose $r=h-g-j-1$ and pre-multiply both sides of the inequality above by $Z^{k-h}$. We obtain that (\ref{eq:convexity_propagating}) implies the following set of inequalities by Lemma~\ref{le:binaries}:
\begin{equation}
\label{eq:intermediate_step}
Z^{k-h}S\Delta_g Z^{h-g-j-1}S \leq Z^{k-j} S\,,
\end{equation}
for all $k \in \mathbb{N}_{[1,N-1]}$, $j \in \mathbb{N}_{[0,k-1]}$, $h \in \mathbb{N}_{[j+1,k]}$ and $g \in \mathbb{N}_{[0,h-j-1]}$.  By definition of the information structure under consideration, we have $S_{k,h}=Z^{k-h}S$, $S_{h-g-1,j}=Z^{h-g-j-1}S$ and $S_{k,j}=Z^{k-j} S$. Hence, (\ref{eq:intermediate_step}) is equivalent to (\ref{eq:binary_small}), which is equivalent to QI by Theorem~\ref{th:characterization_general}. Sufficiency is thus proved.

	\subsection{Proof of necessity}
Let $\mathbf{L},\mathbf{L}' \in \mathbfcal{S}$ and define $\bm{\varphi}=\mathbf{CBL'}$ and $\mathbf{\Phi}=\mathbf{L}\bm{\varphi}=\mathbf{LCBL'}$. Let $\varphi_{i,i-j}$ and $\Phi_{i,i-j}$ denote their $p \times p$ and $m \times p$ block sub-matrices located at block-row $i$ and block-column $i-j$.  Similar to (\ref{eq:binary_small_sums}) we have
\begin{equation}
\label{eq:subterms_diagonal}
\Phi_{i,i-j}=\sum_{h=i-j+1}^{i}L_{i,h}\sum_{g=0}^{h-(i-j)-1}CA^gBL'_{h -g -1, i-j}\,,
\end{equation}
for every $i\in \mathbb{N}_{[0,N-1]}$, $j\in \mathbb{N}_{[1,i]}$. The case $j\leq 0$ is not considered for causality. QI requires that $\mathbf{\Phi} \in \mathbfcal{S}$ for arbitrary $\mathbf{L},\mathbf{L}' \in \mathbfcal{S}$. We proceed into proving necessity of (\ref{eq:convexity_propagating}) by contrapositive. % and show that the negation of (\ref{eq:condition_lemma1}) implies the negation of (\ref{eq:th_QI_requiredCondition}).
	 Let us suppose that (\ref{eq:convexity_propagating}) is violated. Take $g^\star, r^\star$ inside the ranges defined in (\ref{eq:convexity_propagating})  such that
\begin{equation*}
S\Delta_{g^\star}Z^{r^\star}S \not\leq Z^{g^\star+r^\star+1}S\,.
\end{equation*}	
%
%	 Using the matrix product definition and  by Lemma~\ref{le:inclusion}, (\ref{eq:violation_1}) is equivalent to
%	 \small
%	\begin{equation}
%	\label{eq:violation_2}
%	\begin{aligned}
%	&\exists \overline{i},\overline{j},\overline{l},\overline{m},\overline{n} \text{ such that}\\
%	&S(\overline{i},\overline{j})\Delta_{p^\star}(\overline{j},\overline{l})Z^{r^\star}(\overline{l},\overline{m})S(\overline{m},\overline{n})\left(1-\left(Z^{p^\star+r^\star+1}S\right)(\overline{i},\overline{n})\right)=1\,.
%	\end{aligned}
%	\end{equation}
	%We  now find $\mathbf{L}^\star, \mathbf{L'}^\star \in \mathbfcal{S}$ such that $\mathbf{\Phi}^\star=\mathbf{L}^\star\mathbf{CBL'}^\star \notin \mathbfcal{S}$.  
	We find matrices $\mathbf{L,L'} \in \mathbfcal{S}$ and  indices  such that there is a block of $\mathbf{\Phi}$ located at block-row $i$ and block-column $i-j$ for some $i,j$ for which $\text{Struct}(\Phi)_{i,i-j}\not \leq Z^jS$. % so that (\ref{eq:th_QI_requiredCondition}) is also violated.  
%	Consider such $\mathbf{L}, \mathbf{L'}$ to be unknown, and compute the block matrix:
%	\begin{equation*}
%	\label{eq:nice_sum}
%	\Phi_{g^\star+r^\star+1,0}=\sum_{h=1}^{g^\star+r^\star+1}L_{g^\star+r^\star+1,h}\sum_{g=0}^{h-1}CA^{g^\star}BL'_{h -g^\star -1, 0}
%	\end{equation*}
	First, let us pose $L_{a,b}=0_{m \times p}$ for all integers  $a,b$ such that $a\neq g^\star+r^\star+1$ or $b \neq g^\star+r^\star+1$  and $L'_{a,b}=0_{m \times p}$ for all integers $a,b$ such that $a \neq r^\star$ or $b \neq 0$. Then we have
	\begin{equation*}
	\begin{aligned}
		\Phi_{g^\star+r^\star+1,0}&=\sum_{h=1}^{g^\star+r^\star+1}L_{g^\star+r^\star+1,h}\sum_{g=0}^{h-1}CA^gBL'_{h -g -1, 0}\\
		&=L_{g^\star +r^\star +1,g^\star +r^\star +1}CA^{g^\star}BL'_{r^\star,0}\,.
		\end{aligned}
	\end{equation*}
	Finally, we choose matrices $L_{g^\star+r^\star+1,g^\star+r^\star+1}$ and $L'_{r^\star,0}$ such that $\text{Struct}(L_{g^\star+r^\star+1,g^\star+r^\star+1})\leq S$, $\text{Struct}(L'_{r^\star,0})\leq Z^{r^\star}S$ and $$\text{Struct}(L_{g^\star +r^\star +1,g^\star +r^\star +1}CA^{g^\star}BL'_{r^\star,0})=S\Delta_{g^\star}Z^{r^\star}S\,.$$ This can be done by making sure that whenever an entry of $S\Delta_{g^\star}Z^{r^\star}S$ is one, the corresponding entry of $L_{g^\star +r^\star +1,g^\star +r^\star +1}CA^{g^\star}BL'_{r^\star,0}$ is different from zero by avoiding finitely many specific values in the entries of $L_{g^\star +r^\star +1,g^\star +r^\star +1}$ and $L'_{r^\star,0}$. Thus, there are infinitely many suitable choices for $L_{g^\star +r^\star +1,g^\star +r^\star +1}$ and $L'_{r^\star,0}$.	As a result, $\text{Struct}\left(\Phi_{g^\star+r^\star+1,0}\right)=S\Delta_{g^\star}Z^{r^\star}S \not  \leq Z^{g^\star+r^\star +1}$, and $\mathbf{\Phi} \not \in \mathbfcal{S}$, thus concluding the proof. \QEDA

	\bibliographystyle{IEEEtran}
	\bibliography{IEEEabrv,references}
	
%	\begin{IEEEbiography}
%[{\includegraphics[width=1in,height=1.25in,clip]{LF.jpg}}]{Luca Furieri}
%received his B.Sc. and M.Sc. degrees in Automation Engineering  from the University of Bologna, Italy, in 2014 and 2016 respectively. He is currently a Ph.D. candidate at the Automatic Control Laboratory at ETH Z{\"u}rich, Switzerland. His research interests include distributed control and optimization for large-scale systems with applications in power grid systems.
%\end{IEEEbiography}
%
%\begin{IEEEbiography}
%   [ {\includegraphics[width=1.3in,height=1.25in,clip,keepaspectratio]{MK.png}}]{Maryam Kamgarpour}
%is an assistant professor at ETH Z{\"u}rich, Automatic Control Laboratory. She obtained her M.S. and Ph.D. in Control Systems from the University of California, Berkeley (2007, 2011) and her Bachelor of Applied Sciences in Systems Design Engineering from University of Waterloo, Canada (2005). Her research is on multi-agent decision making and control, game theory, hybrid systems and stochastic systems with applications in air traffic, robotics and power grid systems. She is the recipient of NASA High Potential Individual Award, NASA Excellence in Publication Award (2010) and the European Union (ERC) Starting Grant 2015.
%\end{IEEEbiography}
%	
	\end{appendices}
	
	% that's all folks
	\end{document}